\documentclass[10pt,conference,twocolumn,a4paper]{IEEEtran}
\IEEEoverridecommandlockouts

\usepackage[tbtags]{amsmath} 
\usepackage{amssymb}  
\usepackage{enumerate}
\usepackage{amsfonts}
\usepackage[caption=false]{subfig}
\usepackage{color}
\usepackage{cite}

\usepackage{epsfig, graphicx, psfrag}

\usepackage{wasysym}

\usepackage{mathtools}
\mathtoolsset{showonlyrefs=true}

\usepackage{ifthen}
\providecommand{\VersionLength}{long}
\newcommand{\ver}{\ifthenelse{\equal{\VersionLength}{long}}}
\newcommand{\nver}{\ifthenelse{\equal{\VersionLength}{short}}}

\usepackage{amsthm}
\theoremstyle{plain}

\theoremstyle{plain}
\newtheorem{thm}{Theorem}

\theoremstyle{definition}
\newtheorem{defn}{Definition}

\theoremstyle{remark}

\newtheorem*{scheme*}{Scheme}

\newtheorem*{protocol*}{Protocol}

\providecommand{\thmref}[1]{Theorem~\ref{#1}}

\providecommand{\defnref}[1]{Definition~\ref{#1}}
\providecommand{\secref}[1]{Section~\ref{#1}}

\providecommand{\figref}[1]{Fig.~\ref{#1}}

\newcommand{\ie}{i.e.}
\newcommand{\eg}{e.g.}

\newcommand{\ints}{\mathbb{Z}}

\newcommand{\nats}{\mathbb{N}}

\newcommand{\bm}[1]{\mbox{\boldmath{$#1$}}}

\newcommand{\e}{\text{e}}

\newcommand{\mA}{\mathcal{Z}}

\newcommand{\Comment}[1]{}
\newcommand{\old}[1]{}
\newcommand{\rem}[1]{}

\newcommand{\bx}{{\bm x}}

\newcommand{\bv}{{\bm v}}
\newcommand{\bz}{{\bm z}}

\providecommand{\bG}{{\bf G}}

\providecommand{\bK}{{\bf K}}

\providecommand{\e}{{\rm e}}
\providecommand{\comment}[1]{}

\newcommand{\beqn}[1]{\begin{eqnarray}\label{#1}}
\newcommand{\eeqn}{\end{eqnarray}}
\newcommand{\beq}[1]{\begin{equation}\label{#1}}
\newcommand{\eeq}{\end{equation}}

\providecommand{\eps}{\epsilon}

\makeatletter
\newcommand{\vast}{\bBigg@{4}}
\newcommand{\Vast}{\bBigg@{4.9}}
\makeatother

\providecommand{\CC}{convolutional code}
\providecommand{\CCs}{convolutional codes}
\providecommand{\conlen}{d}
\providecommand{\kcc}{k}
\providecommand{\ncc}{n}
\providecommand{\framelen}{N}
\providecommand{\Pe}{\bar{P}_e}
\providecommand{\bias}{B}
\providecommand{\bb}{{\bm b}}
\providecommand{\bc}{{\bm c}}
\providecommand{\bu}{{\bm u}}
\providecommand{\bv}{{\bm v}}
\providecommand{\bw}{{\bm w}}

\providecommand{\tbc}{\tilde{\bc}}
\providecommand{\tbb}{\tilde{\bb}}
\providecommand{\hbb}{\hat{\bb}}

\usepackage{graphicx}
\usepackage{ifthen}

\renewcommand{\VersionLength}{long}

\usepackage{hyperref}

\begin{document}
\title{(Almost) Practical Tree Codes}

\author{Anatoly Khina, Wael Halbawi and Babak Hassibi 
    \\	Department of Electrical Engineering, California Institute of Technology, Pasadena, CA 91125, USA \\
	{\em \{khina, whalbawi, hassibi\}@caltech.edu }
	\thanks{This work was supported in part by the National Science Foundation under grants CNS-0932428, CCF-1018927, CCF-1423663 and CCF-1409204, by a grant from Qualcomm Inc., by NASA's Jet Propulsion Laboratory through the President and Director’s Fund, and by King Abdullah University of Science and Technology.}
}

\maketitle


\begin{abstract}
    We consider the problem of stabilizing an unstable plant driven by bounded noise over a digital noisy communication link, 
a scenario at the heart of 
networked control. 
To stabilize such a plant, 
one needs real-time encoding and decoding with an error probability profile that decays exponentially with the decoding delay. 
The works of Schulman and Sahai over the past two decades have developed the notions of \emph{tree codes} and \emph{anytime capacity},
and provided the theoretical framework for studying such problems. 
Nonetheless, there has been little practical progress in this area due to 
the absence of explicit constructions of tree codes with efficient encoding and decoding algorithms. 
Recently, linear time-invariant tree codes were proposed to achieve the desired result under maximum-likelihood decoding. 
In this work, we take one more step towards practicality, 
by showing that these codes can be efficiently decoded using sequential decoding algorithms, 
up to some loss in performance (and with some practical complexity caveats). 
We supplement our theoretical results with numerical simulations that demonstrate the effectiveness of the decoder in a control system setting.
\end{abstract}


\allowdisplaybreaks

\section{Introduction}
\label{s:intro}

Control theory deals with stabilizing and regulating the behavior of a dynamical system (``plant'') via real-time causal feedback.
Traditional control theory was mainly concerned and used in well-crafted closed engineering systems, 
which are characterized by the measurement and control modules being co-located. 
The theory and practice for this setup are now well established; see, \eg, \cite{Babook:Indefinite}.

Nevertheless, in the current technological era of ubiquitous wireless connectivity, the demand for control over wireless media is ever growing. This \emph{networked control} setup presents more challenges due to its distributed nature: 
The plant output and the controller are no longer co-located and are separated by an unreliable link (see \figref{fig:noisy_feedback_link}).

To stabilize an unstable plant using the unreliable feedback link, an error-correcting code needs to be employed over the latter.
In one-way communications~--- the cornerstone of information theory~--- 
all the source data are assumed to be known in advance (\emph{non-causally}) and are recovered only when the reception ends.
In contrast, in coding for control, the source data are known only causally, as the new data at each time instant are dependent upon 
the dynamical random process. 
Moreover, the controller cannot wait until a large block is received; it needs to constantly produce estimates of the system's state, 
such that the fidelity of earlier data improves as time advances. Both of these goals are achieved via causal coding, which receives the data sequentially in a causal fashion and encodes it in a way such that the error probability of recovering the source data at a fixed time instant improves constantly with the reception of more code symbols.
\\%
\indent Sahai and Mitter \cite{SahaiMitterPartI} provided necessary and sufficient 
conditions on the required communication reliability 
over the unreliable feedback link to the controller. 
To that end, they defined the notion of \emph{anytime capacity} as the appropriate figure of merit for this setting, 
which is essentially the maximal possible rate of a causal code that at \emph{any} time $t$
recovers a source bit at time $(t - d)$ with error probability that decays 
exponentially with $d$, for \emph{all} $d$.
They further recognized that such codes have a natural \emph{tree code} structure, 
which is similar to the codes developed by Schulman for the related problem of interactive computation~\cite{TreeCodes}.

Unfortunately, the result by Schulman (and consequently also the ones by Sahai and Mitter) only proves the existence of a tree code with the desired properties and does not guarantee that a random tree code would be good with high probability. 
The main difficulty comes from the fact that proving that the \emph{random ensemble} achieves the desired exponential decay does not guarantee that the \emph{same code} achieves this for \emph{every time instant} and \emph{every delay}.

Sukhavasi and Hassibi \cite{SukhavasiHassibi} circumvented this problem by introducing linear time-invariant (LTI) tree codes. 
The time-invariance property means that the behavior of the code at every time instant is the same, 
which suggests, in  turn, that the performance guarantees for a \emph{random (time-invariant) ensemble} are easily translated 
to similar guarantees for a \emph{specific code} chosen at random, \emph{with high probability}.

However, this result assumes maximum likelihood (ML) decoding, which is impractical except for binary erasure channels (in which case it amounts to solving linear equations which has polynomial computational complexity).
\begin{figure}[t]
    \vspace{.6\baselineskip}
    \includegraphics[scale=1]{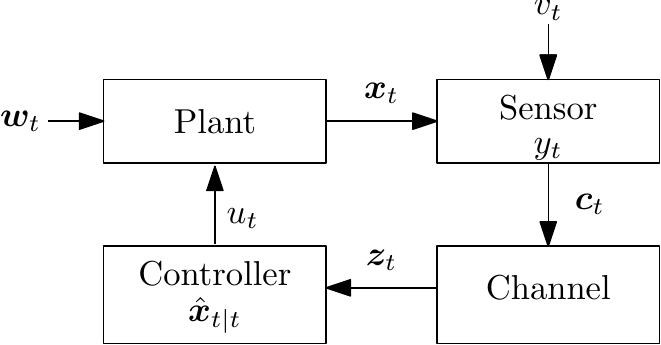}
    \centering
    \caption{Basic networked control system.}
    \label{fig:noisy_feedback_link}
\end{figure}

Sequential decoding was proposed by Wozencraft~\cite{WozencraftPhD} and subsequently improved by others as a means to recover random tree codes with reduced complexity with some compromise in performance; specifically, for the expected complexity to be finite, the maximal communication rate should be lower than the cutoff rate. For a thorough account of sequential decoding, see~\cite[Ch.~10]{JelinekBook}, \cite[Sec.~6.9]{GallagerBook1968}, \cite[Ch.~6]{ViterbiOmuraBook}, \cite[Ch.~6]{ZigangirovBook}.
This technique was subsequently adopted by Sahai and Palaiyanur~\cite{Sahai_SequentialDecoding} for the purpose of decoding (time-varying) tree codes for networked control. 
Unfortunately, this result relies on an exponential bound on the error probability by Jelinek \cite[Th.~2]{Jelink_SequentialDecoding:Parameters} that is valid 
for the binary symmetric channel (BSC) (and other cases of interest) only when the expected complexity of the sequential decoder goes to infinity~\cite{ArikanSequentialDecoding}.

In this work we propose the usage of sequential decoding for the recovery of LTI tree codes. 
To that end, similarly to Sahai and Palaiyanur~\cite{Sahai_SequentialDecoding}, 
we extend a (different) result developed by Jelinek~\cite[Th.~10.2]{JelinekBook} for general (non-linear and time-variant) random codes to LTI tree codes.


\section{Problem Setup and Motivation}
\label{s:bkg}

We are interested in stabilizing an unstable plant driven by bounded noise over a noisy communication link. In particular, an observer of the plant measures at every time instant $t$ a noisy version $y_t \in \mathbb{R}$ (with bounded noise) of the state of the plant $\bx_t \in \mathbb{R}^m$. The observer then quantizes $y_t$ to $\bb_t \in \mathbb{Z}_2^k$, and encodes --- using a causal code --- all quantized measurements $\{\bb_i\}_{i = 1}^{t}$ to produce $\bc_t \in \mathbb{Z}_2^n$.  This packet $\bc_t$ is transmitted over a noisy communication link to the controller, which receives $\bz_t \in \mA^n$,  where $\mA$ is the channel output alphabet. The controller then decodes $\{\bz_i\}_{i = 1}^{t}$ to produce the estimates $\{\hat{\bb}_{i|t}\}_{i = 1}^{t}$, 
where $\hat{\bb}_{i|t}$ denotes the estimate of $\bb_i$ when decoded at time $t$. These estimates are 
mapped back to measurement estimates $\{\hat{y}_{i|t}\}_{i = 1}^{t}$ which, in turn, are used to give an estimate $\hat{\bx}_{t|t}$ of the current state of the plant. Finally, the controller computes a control signal $\bu_{t}$ based on $\hat{\bx}_{t|t}$ and applies it to the plant. 

The need for causally sending measurements of the state in real time motivates the use of causal codes in this problem. Generally speaking, a causal code maps, at each time instant $t$, the current and all previous quantized measurements to a packet of $n$ bits$\bc_t$,
\begin{align}
\bc_t = f_{t}\left(\left\{\bb_i\right\}_{i = 1}^t\right) .
\end{align}

When restricted to linear codes, each function 
$f_t$
can be characterized by a set of matrices $\{\mathbf{G}_{t,1},\ldots,\mathbf{G}_{t,t}\}$, where $\mathbf{G}_{t,i} \in \mathbb{Z}_2^{n \times k}$. 
The sequence of quantized measurements at time $t$, $\{\bb_i\}_{i=1}^t$, is encoded as,
\begin{align}
\bc_t = \mathbf{G}_{t,1}\bb_1 + \mathbf{G}_{t,2}\bb_2 + \cdots +  \mathbf{G}_{t,t}\bb_t \,.
\end{align}
The decoder computes a function $g_t\left(\{\bz_{i}\}_{i = 1}^t\right)$ to produce $\{\bb_{i|t}\}_{i = 1}^{t}$. 
One is then assigned the task of choosing a sequence of matrices 
$\{\mathbf{G}_{t,i} \, | \, i = 1, \ldots, t \,; \, t=1, \ldots, \infty \}$ 
that provides anytime reliability. We recall this definition as stated in~\cite{SukhavasiHassibi}. 

\begin{defn}
\label{def:anytime_reliability:system}
    Define the probability of the first error event as 
    \begin{align*}
	P_e(t, d) \triangleq 
	P \left( b_{t-d} \neq \hbb_{t-d|t}, \forall \delta > d, b_{t-\delta} = \hbb_{t-\delta|t} \right), 
    \end{align*}
    where the probability is over the randomness of the plant and the channel noise. Suppose we are assigned a budget of $n$ channel uses per time step of the evolution of the plant. 
    Then, an encoder--decoder pair is called $(R,\beta)$ anytime reliable if there exists $d_0 \in \nats$, such that
    \begin{align}
    \label{eq:anytime_reliability}
	P_e(t,d) &\leq 2^{-\beta nd}, & \forall t,d \geq d_0 , 
    \end{align}
    where $\beta$ is called the anytime exponent. 
\end{defn}
According to the definition, anytime reliability has to hold for every decoding instant $t$ and every delay $d$. Sukhavasi and Hassibi proposed in~\cite{SukhavasiHassibi} a code construction based on Toeplitz block-lower triangular parity-check matrices, that provides an error exponent for all $t$ and $d$. 
The Toeplitz property, in turn,
avoids 
the need to compute a double union bound. We shall explicitly show this later in \secref{ss:LTI}, where we introduce the LTI ensemble.

In this work, 
the communication
link between the observer and the controller is assumed to be a memoryless binary-input output-symmetric (MBIOS) channel: $w(z_i|c_i = 0) = w(-z_i|c_i = 1)$, where $w$ is the channel transition distribution, $z_i \in \mA$ and $c_i \in \mathbb{Z}_2$.


\section{Preliminaries: Convolutional Codes}
\label{s:CC}

In this section we
review known result for several random ensembles of \CCs, in \secref{ss:CC:bounds}.
The codes within each ensemble can be either linear or not; linear ensembles can be further either time variant 
or time invariant. We further discuss sequential decoding algorithms and their performance 
in \secref{ss:CC:ML} 
which will be applied in the sequel for tree codes.


\subsection{Bounds on the Error Probability under ML Decoding}
\label{ss:CC:bounds}

We now recall exponential bounds for convolutional codes under certain decoding regimes.

A compact representation (and implementation)
of a \CC\ is via a shift register:
The delay-line (shift register) length is denoted by $\conlen$, 
whereas its width $\kcc$ is the number of information bits
entering the shift register at each stage. Thus, the total
memory size is equal to $\conlen \kcc$ bits. At each stage, $\ncc$ code
bits are generated by evaluating $\ncc$ functionals over the $\conlen \kcc$
memory bits and the new $\kcc$ information bits. We refer to these $\ncc$ bits 
as a single branch.
Therefore, the rate of the code is equal to
$R = \kcc / \ncc$
bits per channel use. 
In general, these functionals may be either linear or not, resulting in linear or non-linear \CCs, respectively, 
and stay fixed or vary across time, resulting in time-invariant or time-variant \CCs.
We further denote the total length of the \CC\ frame upon truncation by $\framelen$.

Typically, the total length of the \CC\ frame is chosen to be much larger than $\conlen$, \ie, $\framelen \gg \conlen$.
We shall see in \secref{ss:LTI}, that in the context of tree codes, 
a decoding delay of $\conlen$ time steps of the evolution of the plant into the past corresponds to a \CC\ with delay-length $\conlen$.
Since each time step corresponds to $\ncc$ uses of the communication link, 
the relevant regime for the current work is 
$\framelen = \ncc \conlen$.

\begin{thm}[{\!\cite[Sec.~5.6]{ViterbiOmuraBook}, \cite[Sec.~4.8]{ZigangirovBook}\!}]
\label{thm:CC:ML:BER}
    The probability of the first error event of random time-variant and LTI \CCs, under optimal (maximum likelihood) decoding, is bounded from the above by
    \begin{align}
    \label{eq:CC:Pe}
	\Pe(d) \leq 2^{- E_G(R) \ncc \conlen}
    \end{align}
    where $E_G(R)$ is Gallager's error exponent function~\cite[Sec.~5.6]{GallagerBook1968}, defined as (see also \figref{fig:exponents}):
    \begin{subequations}
    \label{eq:Er}
    \noeqref{eq:Er:def,eq:Er:E0}
    \begin{align}
	E_G(R) &\triangleq \max_{0 \leq \rho \leq 1} \left[ E_0(\rho) - \rho R \right] ,
    \label{eq:Er:def}
     \\
     E_0(\rho) &\triangleq 1+\rho - \log \left\{ \sum_{z \in \mA} \left[ w^{\frac{1}{1+\rho}}(z|0) + w^{\frac{1}{1+\rho}}(z|1) \right]^{1 + \rho} \right\} .
    \label{eq:Er:E0}
    \end{align}
    \end{subequations}
\end{thm}

    We note that 
    in the common work regime of $\framelen \gg \conlen$, 
    the optimal achievable error exponent was proved by Yudkin and by Viterbi to be much better than $E_G(R)$ \cite[Ch.~5]{ViterbiOmuraBook}. 
    Unfortunately, this result does not hold for the case of $\framelen = \ncc \conlen$ which is the relevant regime for this work.

Interestingly, whereas time-variant codes are known to achieve better error exponents than linear time-invariant (LTI) ones when $\framelen \gg \conlen$, this gain vanishes when $\framelen = \ncc \conlen$, as is suggested by the following theorem.
\begin{thm}[{\cite[Eq.~(14)]{TimeInvariantConvolutional}}]
    The probability of the first error event of LTI random \CCs, under optimal (maximum likelihood) decoding, is bounded from the above by \eqref{eq:CC:Pe}.
\end{thm}

Thus, \eqref{eq:CC:Pe} remains valid for LTI codes.

Unfortunately, the computational complexity of maximum-likelihood decoding grows exponentially with the delay-line length $\conlen$, prohibiting its use in practice for large values of $\conlen$.

We therefore review next a suboptimal decoding procedure, the complexity of which does not grow rapidly with $\conlen$ but still achieves exponential decay in $\conlen$ of the BER.


\subsection{Sequential Decoding}
\label{ss:CC:ML}

The Viterbi algorithm~\cite[Sec.~4.2]{ViterbiOmuraBook} offers an efficient implementation of (frame-wise) ML\footnote{For bitwise ML decoding, the BCJR algorithm~\cite{BCJR} needs to be used.} decoding for fixed $d$ and growing $\framelen$.
Unfortunately, the complexity of this algorithm grows exponentially with $\framelen$ when the two are coupled, 
\ie, $\framelen = \ncc \conlen$.\footnote{This is true with the exception of the binary erasure channel, for which ML decoding amounts to solving a system of equations, the complexity of which is polynomial.}
Prior to the adaptation of the Viterbi algorithm as the preferred decoding algorithm of \CCs, 
sequential decoding was served as the \emph{de facto} standard. 
A wide class of algorithms fall under the umbrella of ``sequential decoding''. 
Common to all is the fact that they explore only a subset of the (likely) codeword paths, such that 
their complexity does not grow (much) with $d$, and are therefore applicable for the decoding of tree codes.\footnote{Interestingly, the idea of tree codes was conceived and used already in the early works on sequential decoding~\cite{WozencraftPhD}. These codes were used primarily for the classical communication problem, and not for interactive communication or control.}

In this work we shall concentrate on the two popular variants of this algorithm~--- the Stack and the Fano (which is characterized by a quantization parameter $\Delta$) algorithms.

We next summarize the relevant properties of these decoding algorithms when using the generalized Fano metric (see, \eg, \cite[Ch.~10]{JelinekBook}) to compare possible codeword paths: 
\begin{subequations}
\label{eq:FanoMetric}
\noeqref{eq:FanoMetric:total,eq:FanoMetric:single}
\begin{align}
    M(c_1, \ldots, c_N) &= \sum_{t=1}^N M(c_t) ,
 \label{eq:FanoMetric:total}
      \\
M(c_t) &\triangleq \log \frac{w(z_t|c_t)}{p(z_t)} - \bias ,
 \label{eq:FanoMetric:single} 
\end{align}
\end{subequations}
where $\bias$ is referred to as the metric bias and penalizes longer paths when the metrics of different-length paths are compared.
In contrast to ML decoding, where all possible paths (of length $N$) are explored to determine the path with the total maximal metric,\footnote{Note that optimizing \eqref{eq:FanoMetric:total} in this case is equivalent to ML decoding.} using the stack sequential decoding algorithm, 
a list of partially explored paths is stored, 
where at each step the path with the highest metric is further explored and replaced with its immediate descendants and their metrics. The Fano algorithm achieves the same without storing all these potential paths, at the price of a constant increase in the error probability and computational complexity; for a detailed descriptions of both algorithms see~\cite[Ch.~10]{JelinekBook}, \cite[Sec.~6.9]{GallagerBook1968}, \cite[Ch.~6]{ViterbiOmuraBook}, \cite[Ch.~6]{ZigangirovBook}.

The choice $\bias = R$ is known to minimize the expected computational complexity, and is therefore the most popular choice in practice. Moreover, for rates below the cutoff rate $R < R_0 \triangleq E_0(\rho=1)$, 
the expected number of metric evaluations \eqref{eq:FanoMetric:single} is finite and does not depend on $d$, for any $\bias \leq R_0$~\cite[Sec.~6.9]{GallagerBook1968}, \cite[Ch.~10]{JelinekBook}.
Thus, the only increase in expected complexity of this algorithm with $d$ comes from an increase in the complexity of evaluating the metric of a single symbol \eqref{eq:FanoMetric:single}. 
Since the latter increases (at most) linearly with $d$, 
the total complexity of the algorithm grows polynomially in $d$. 
Furthermore, for rates above the cutoff rate, $R > R_0$, 
the expected complexity is known to grow rapidly with $N$ for \emph{any metric}~\cite{ArikanSequentialDecoding}, 
implying that the algorithm is applicable only for rates below the cutoff rate.

Most results concerning the error probability under sequential decoding consider an infinite code stream (over a trellis graph) and evaluate the probability of an erroneous path to diverge from the correct path and re-merge with the correct path, which can only happen for $\framelen > \ncc \conlen$.
Such analyses are not adequate for our case of interest, in which $\framelen = \ncc \conlen$. 
The following theorem provides a bound for our case.

\begin{thm}[{\!\! \cite[Ch.~10]{JelinekBook}}]
\label{thm:Jelinek:RandomCC}
    The probability of the first error event of general random \CCs, 
    using the Fano or stack sequential decoders and the Fano metric with bias $\bias$, 
    is bounded from the above by:
    \begin{subequations}
    \label{eq:SeqDec:Pe}
    \noeqref{eq:SeqDec:Pe:bound,eq:SeqDec:Pe:exponent}
    \begin{align}
	\Pe(t, \conlen) &\leq A\,2^{ - E_J(\bias,R) \ncc \conlen} , 
    \label{eq:SeqDec:Pe:bound}
     \\ E_J(\bias, R) &\triangleq \max_{0 \leq \rho \leq 1} \frac{\rho}{1 + \rho} \Big\{ E_0(\rho) + \bias - (1 + \rho) R \Big\} ,
    \label{eq:SeqDec:Pe:bound}
    \end{align}
    \end{subequations}
    where $A$ is finite for $\bias < R_0$ and is upper bounded by\footnote{Note that $E_0(\rho)/\rho$ is a monotonically decreasing function of $\rho$, therefore $\bias < R_0 = E_0(1)$ guarantees that $E_0(\rho) - \rho \bias > 0$.}
    \begin{align}
    \label{eq:SeqDec:A}
	A \leq \e^{\frac{\rho}{1+\rho} \Delta} \frac{1 - \e^{-t [E_0(\rho) - \rho \bias]}}{1 - \e^{-[E_0(\rho) - \rho \bias]}}
	\leq \frac{\e^{\frac{\rho}{1+\rho} \Delta} }{1 - \e^{-[E_0(\rho) - \rho \bias]}} < \infty ,\ \  
    \end{align}
    for a quantization parameter $\Delta$ in the Fano algorithm; for the stack algorithm \eqref{eq:SeqDec:A} holds true with 
    $\Delta = 0$.
\end{thm}

Since, $E_J(\bias,R)$ is a monotonically increasing function of $\bias$, 
choosing $\bias = R_0$ maximizes the exponential decay of $\Pe(\conlen)$ in $\conlen$.\footnote{For finite values of  $\conlen$ a lower choice of $\bias$ might be better, since the constant $A$ might be smaller in this case.}
Interestingly, for this choice of bias, $E_J(\bias=R_0, R) = E_G(R)$ whenever $E_G(R)$ is achieved by $\rho = 1$ 
in \eqref{eq:Er:def}, \ie, for rates below the critical rate $R < R_\text{crit} \triangleq E'_0(\rho)$. For other values of $\rho$, $E_J(\bias = R_0, R)$ is strictly smaller than $E_G(R)$ (see \figref{fig:exponents}).

For the choice of bias that optimizes complexity, $\bias = R$, on the other hand, an error exponent which equals to at least half the exponent under ML decoding \eqref{eq:CC:Pe} is achieved whenever $E_G(R)$ is achieved by $\rho=1$: $E_J(\bias=R,R) \geq E_G(R) / 2$ (see \figref{fig:exponents}).

\begin{figure}[t]
    \centering
    \epsfig{file = 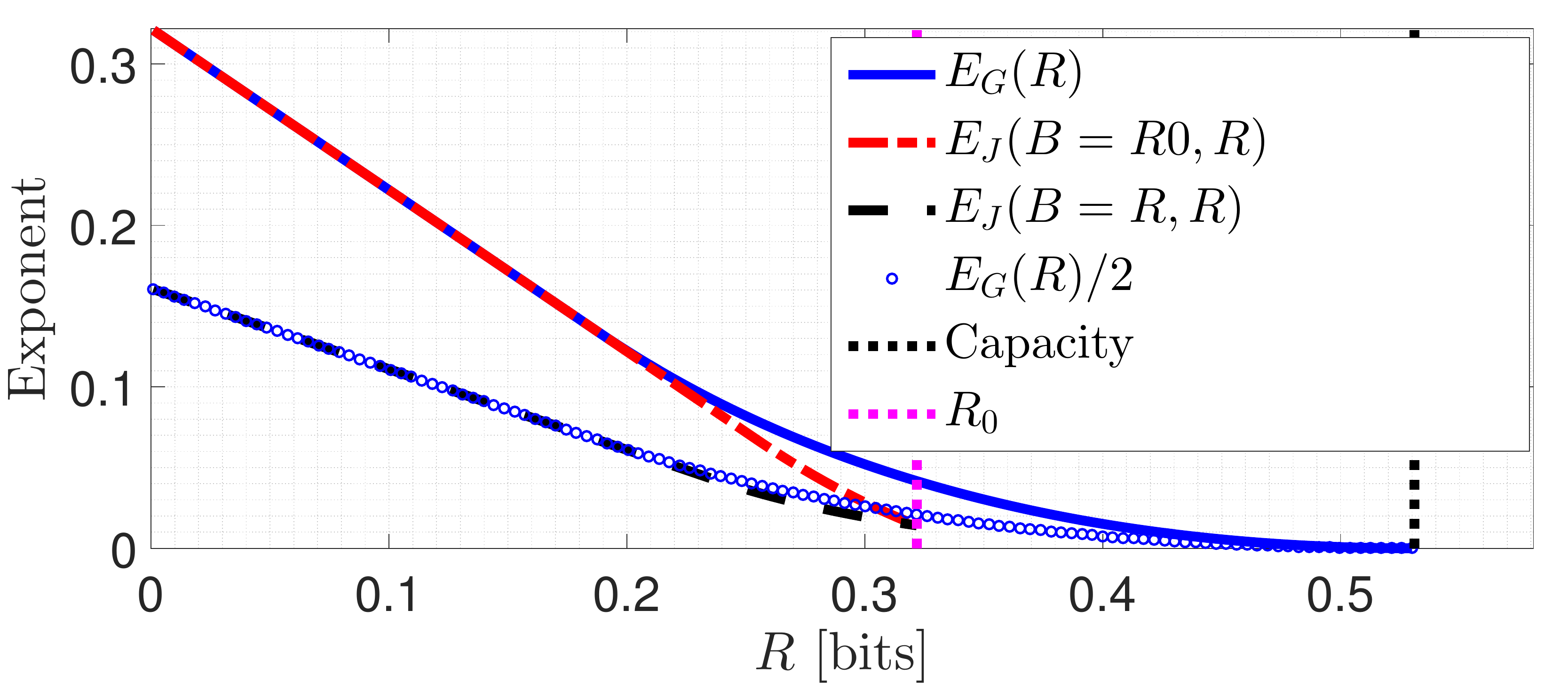, width = \columnwidth}
    \centering
    \caption{Error exponents $E_G(R)$, $E_J(\bias=R_0, R)$ and $E_J(\bias=R, R)$ for a BSC with crossover probability $p = 0.1$.}
    \label{fig:exponents}
\end{figure}

We now turn to bounding the number of branch computations per node of the code tree under sequential decoding.

\begin{defn}
\label{defn:complexity}
    Denote by $W_t$ the number of branch computations of node $t$ performed by a sequential decoding algorithm.
\end{defn}

We note that $W_t$ is a random variable (which depends on the received vector and the underlying tree code used). 
Since $W_t$ is equal to just one more than the number of branch computations in the incorrect sub-tree of that node, 
it has the same distribution for any $t$, for general random and LTI random codes.

We use the following result which can be found in \cite[Sec.~6.2]{ViterbiOmuraBook} for general random \CCs, 
and in~\cite[Sec.~6.9]{GallagerBook1968} for LTI \CCs; 
for both, delay-line length $\conlen$ is assumed to be infinite.\footnote{For finite delay-length \CCs, the computational complexity can only be smaller than that of infinite delay-length codes.}

\begin{thm}
\label{thm:SeqDec:complexity:UB}
    The probability that $W_t$ of a general random \CC\ or an LTI \CC\ with infinite delay-length is larger than $m \in \nats$ is upper bounded by 
    \begin{align}
    \label{eq:Pareto:UB}
	\Pr \left( W_t \geq m \right) \leq A m^{-\rho} ,
    \end{align}
    where $A$ is finite for $\bias, R < R_0$, $R < \frac{\bias + R_0}{2 \rho}$, and $\rho \in (0, 1]$.
\end{thm}
An immediate consequence of this result is that the expected complexity per branch $E[W_t]$ is bounded if $R < R_0$.
Moreover, a converse result by Ar{\i}kan~\cite{ArikanSequentialDecoding} states that the expected complexity is unbounded for rates exceeding the cutoff.

The next result provides a lower bound on the error probability.
\begin{thm}[See {\cite[Sec.~6.4]{ViterbiOmuraBook}}]
\label{thm:SeqDec:complexity:LB}
    The probability that $W_t$ of any \CC, where no decoding error occurs, is greater than $m \in \nats$, is lower bounded by 
    \begin{align}
    \label{eq:Pareto:LB}
	\Pr \left( W_t \geq m \right) \geq (1 - o(m)) m^{-\rho} ,
    \end{align}
    where $o(m) \to 0$ for $m \to \infty$, and $R = \frac{E_0(\rho)}{\rho}$ and any $\rho > 0$.\footnote{Recall that $E_0(\rho)/\rho$ is a decreasing function of 
    $\rho$ and therefore $\rho > 1$ implies that $R < R_0$.} 
\end{thm}
This result was proved to be tight by Savage~\cite{Savage_SequentialDecoding} for \emph{general} random \CC, and is widely believed to be true for random LTI \CCs, although no formal proof exists for the latter.


\section{Linear Time Invariant Anytime Reliable Codes}
\label{ss:LTI}

In this section, we recall the construction of LTI anytime reliable codes as presented in~\cite{SukhavasiHassibi}. A causal linear time-invariant code has a parity-check matrix with the following lower triangular Toeplitz structure. 
 The generator matrix of an $(n, R)$ LTI code of rate $R$ and blocklength $n$ per time step is given by $\bc = \bG_{n,R} \bb$, where
\begin{align}
\label{eqn:toeplitz PC}
    \bG_{n,R} =
    \!
    \begin{bmatrix}
	\bG_1 	& \mathbf{0} 	& \cdots 	& \cdots 	& \cdots \\
	\bG_2 	& \bG_1 	& \mathbf{0} 	& \cdots 	&\cdots \\
	\vdots 	& \vdots 	& \ddots  	& \ddots 	& \cdots \\
	\bG_t 	& \bG_{t - 1} 	& \cdots 	& \bG_1 	& \mathbf{0} \\
	\vdots 	& \vdots 	& \vdots  	&\vdots 	& \ddots
    \end{bmatrix}
    \!\!,
        \bb =
    \! 
        \begin{bmatrix}
            \bb_1
	 \\ \bb_2
	 \\ \vdots
	 \\ \bb_t
	 \\ \vdots
        \end{bmatrix}
    \!\!,
        \bc = 
    \!
        \begin{bmatrix}
            \bc_1
	 \\ \bc_2
	 \\ \vdots
	 \\ \bc_t
	 \\ \vdots
        \end{bmatrix}
        \!\!
\end{align}
and $\bG_t \in \mathbb{Z}_2^{n \times k}$. The following definition, 
given in~\cite{SukhavasiHassibi}, dictates the way to choose the $\bG_i$'s.
\begin{defn}[LTI Ensemble]
\label{def:LTI_ensemble}
    Fix $\bG_1$ to be a full rank matrix, and generate the entries of $\bG_t$ independently and uniformly at random, for $t \geq 2$.
\end{defn}

For the purpose of this paper, we shall view an $(n, R)$ LTI code as a convolutional code with infinite delay-line length, $\kcc$ information bits and $\ncc$ code bits. As a result of this interpretation, the results of Section \ref{s:CC} apply directly to this Toeplitz ensemble.

We shall now show why such a construction does indeed guarantee anytime reliability as defined in \eqref{eq:anytime_reliability}. 

By using the Markov inequality along with the result of \thmref{thm:CC:ML:BER}, 
the probability that a particular code from this ensemble has an exponent that is strictly smaller than $(E_G(R) - \epsilon)$ is
bounded from the above by 
\begin{align}
\label{eq:UB:specific_d}
    P\left( P_e(d) > 2^{-(E_G(R) - \epsilon) \ncc \conlen}\right) \leq 2^{-\eps \ncc \conlen} .
\end{align}
Thus, for any $\epsilon > 0$, this probability can be made arbitrarily small by taking $\conlen$ to be large enough.

However, for a code to be anytime reliable it needs to satisfy \eqref{eq:anytime_reliability} 
for every $t$ and $\conlen_0 \leq \conlen \leq t$.
Unfortunately, applying the union bound and \eqref{eq:UB:specific_d} to 
\begin{align}
    P\left( \bigcup_{t = 1}^{\infty} \bigcup_{\conlen = \conlen_0}^{t} \left\{ P_e(t,\conlen) > 2^{-(E_G(R) - \epsilon) \ncc \conlen}\right\} \right) 
\end{align}
gives a trivial upper bound.

The advantage of using an LTI code is that for a fixed $d$, the event $\left\lbrace P_e(t,d) > 2^{-(E_G(R) - \epsilon)nd} \right\rbrace$ is identical for all $t$.
Therefore, for LTI codes, we have
\begin{subequations}
\label{eq:LTI:UB}
\noeqref{eq:LTI:UB:lhs,eq:LTI:UB:1UB,eq:LTI:UB:series,eq:LTI:UB:final}
\begin{align}
\label{eq:LTI:UB:lhs}
    &P\left( \bigcup_{t = 1}^{\infty}\bigcup_{\conlen = \conlen_0}^{t} \left\lbrace P_e(t,\conlen) > 2^{-(E_G(R) - \epsilon) \ncc \conlen}\right\rbrace \right)\\
\label{eq:LTI:UB:1UB}
    &= P\left(\bigcup_{\conlen = \conlen_0}^{\infty} \left\lbrace P_e(t,d) > 2^{-(E_G(R) - \epsilon) \ncc \conlen}\right\rbrace \right)\\
\label{eq:LTI:UB:series}
    & \leq \sum_{\conlen = \conlen_0}^{\infty} 2^{-\eps nd}\\
\label{eq:LTI:UB:final}
    & = \frac{2^{-\eps \ncc \conlen_0}}{1 - 2^{-\eps \ncc}}
    \,.
\end{align}
\end{subequations}
As a result, a large enough $\conlen_0$ guarantees that a specific code selected at random for the LTI ensemble 
achieves \eqref{eq:anytime_reliability} with exponent $\beta = (E_G(R) - \epsilon)$, for all $t$ and $\conlen_0 \leq \conlen \leq t$, 
with high probability.


\section{Sequential Decoding of Linear Time-Invariant Anytime Reliable Codes}
\label{s:LTI}

In this section we show that the upper bound on the probability of the first error event under sequential decoding for \CCs~\eqref{eq:SeqDec:Pe}, 
holds true also for LTI \CCs, which for $\framelen = \ncc \conlen$, 
identifies with the LTI tree codes of~\eqref{eqn:toeplitz PC}.

To prove this, we adopt the proof technique of \cite[Sec.~6.2]{GallagerBook1968}, where the exponential bounds on the error probability of random block codes are shown to hold also for linear random blocks codes.

\begin{thm}
\label{lem:LTI:SeqDec}
    The probability of the first error event of the LTI random tree ensemble of \defnref{def:LTI_ensemble}, 
    using the Fano or stack sequential decoders and the Fano metric with bias $\bias$, 
    is bounded from the above by \eqref{eq:SeqDec:Pe}.
\end{thm}

\begin{proof}[Proof sketch]
    A thorough inspection of the proof of \thmref{thm:Jelinek:RandomCC}, as it appears in \cite[Ch.~10]{JelinekBook}, 
    reveals that the following two requirements for this bound to be valid are needed:
    \begin{enumerate}
    \item \textbf{Pairwise independence.}
	Every two paths are independent starting from the first branch that corresponds to source branches that disagree.
    \item \textbf{Individual codeword distribution.}
	The entries of each codeword are i.i.d.\ and uniform.
    \end{enumerate}
    We next show how these two requirements are met for the affine linear ensemble.
    The codes in this ensemble are as in \eqref{eqn:toeplitz PC} up to an additive translation 
    $\bv = \begin{bmatrix} \bv_1 & \bv_2 & \cdots & \bv_t \cdots \end{bmatrix}^T$, 
    where $\bv_t \in \ints_2^\ncc$, with $\bc = \bG \bb + \bv$.
    The entries of $\bG$ and $\bv$ are sampled independently and uniformly at random.

    Now, assume that two source words $\bb$ and $\tbb$ are identical for $i < t$ and differ in at least one bit in branch $t$, \ie, $\bb_i = \tbb_i$ for $i < t$ and $\bb_t \neq \tbb_t$.
    Then, the causal structure of $\bG$ guarantees that also $\bc_i = \tbc_i$ for $i < t$.
    Moreover, $\bb_t \neq \tbb_t$ along with the random construction of $\bG$ suggest that the two code paths starting from branch $t$, 
    $\begin{bmatrix}  \bc_t^T &  \bc_{t+1}^T & \cdots \end{bmatrix}^T$ and 
    $\begin{bmatrix} \tbc_t^T & \tbc_{t+1}^T & \cdots \end{bmatrix}^T$
    are independent. This establishes the first requirement.
    
    To establish the second requirement we note that the addition of a random uniform translation vector $\bv$
    guarantees that the entries of each codeword are i.i.d.\ and uniform. This establishes the second requirement and hence also the validity of the proof of \cite[Ch.~10]{JelinekBook}.
    
    Finally, note that since the channel is MBIOS, the same error probability is achieved for any translation vector $\bv$.
\end{proof}
Since 
\thmref{lem:LTI:SeqDec}
holds for LTI codes, 
a specific code chosen from the LTI ensemble is anytime reliable with \eqref{eq:SeqDec:Pe}, with high probability, following \eqref{eq:LTI:UB}.


\section{Simulation of a Control System}
\label{s:numerics}

To demonstrate the effectiveness of the sequential decoder in stabilizing an unstable plant driven by bounded noise, 
we simulate a cart--stick balancer controlled by an actuator that obtains noisy measurements of the state of the cart through a BSC.
The example is the same one from~\cite{SukhavasiHassibi} which is originally from~\cite{FranklinPowellEmamiBook}.
The plant dynamics evolve as,
\begin{align}
    \bx_{t+1} &= \mathbf{A}\bx_t + \mathbf{B}\bu_t + \bw_t \label{eqn:plant_state}
    \\ y_t &= \mathbf{C}\bx_t + v_t ,
\end{align}
where $\bu_t$ is the control input signal that depends only on the estimate of the current state, \ie\ $\bu_t = \bK \hat{\bx}_{t|t}$. 
The system noise $\bw_t$ is a vector of i.i.d. Gaussian random variables with mean $\mu = 0$ and variance $\sigma^2 = 0.01$, truncated to $[-0.025, 0.025]$. The measurement noise $v_t$ is also a truncated Gaussian random variable of the same parameters. We assume that the system is in observer canonical form:
\begin{align}
&\mathbf{A} = \begin{bmatrix}
\phantom{-} 3.3010   &  1 & 0 \\
-3.2750  &  0 & 1 \\
\phantom{-} 0.9801  &  0 & 0 
\end{bmatrix}, \mathbf{B} = \begin{bmatrix}
-0.0300\\
-0.0072\\
\phantom{-} 0.0376
\end{bmatrix},\mathbf{C} = \begin{bmatrix}
1 & 0 & 0
\end{bmatrix}, \\[-.75\baselineskip]
&\mathbf{K} = \begin{bmatrix}
-55.6920  & -32.3333 & -19.0476
\end{bmatrix}
.
\end{align}

The state $\bx_t$, before the transformation to observer canonical form, is composed of the stick's angle, the stick's angular velocity and the cart's velocity. The system is unstable with the largest eigenvalue of $\mathbf{A}$ being 1.75. The channel between the observer and the controller is a BSC with bit-flip probability $p = 0.01$, which has a cutoff rate $R_0 = 0.7382$. We fix a budget of $n = 20$ channel uses. Using Theorem 8.1 in~\cite{SukhavasiHassibi}, the minimum required number of quantization bits is $\kcc_{\text{min}} = 3$ and the minimum required exponent is $\beta_{\min} = 0.2052$. 

For the first experiment, we use a code of rate $R = 1/2$, with $k = 10$ bits for a lattice quantizer with bin width $\delta = 0.1$. From \eqref{eq:SeqDec:Pe:bound}, a sequential decoder with bias $R_0$ will guarantee an error exponent of $\beta = 0.2382$. As is evident from the dark curve in \figref{fig:angle}, the stick on the cart does not deviate by more than 12 degrees.

For the second experiment, we use a lower code rate $R = 1/5$, which provides $\beta = 0.5382$. The light curve in \figref{fig:angle} shows that the deviation is reduced to 4 degrees.

\begin{figure}[t]
\centering
\includegraphics[scale=0.4]{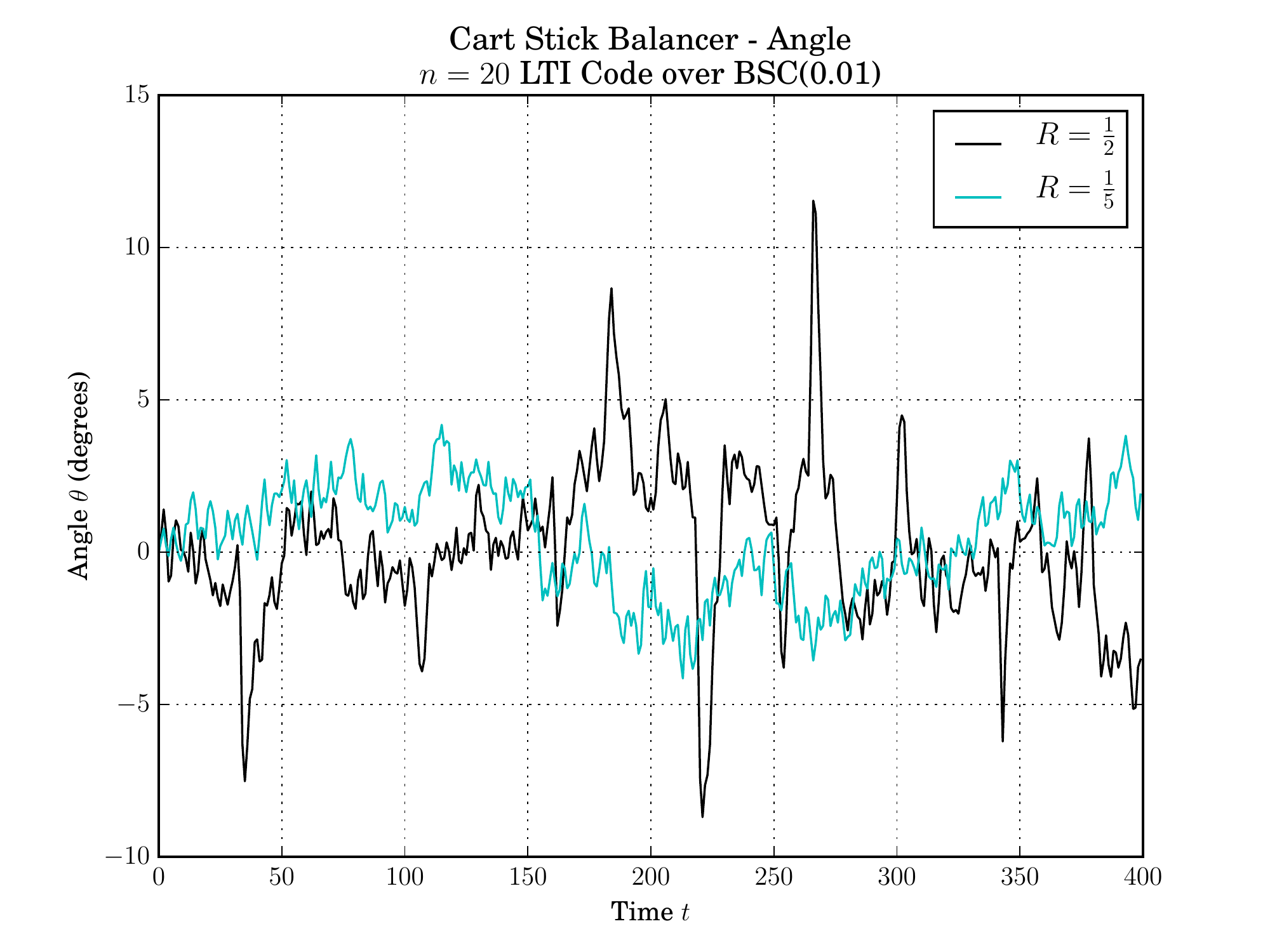}
\caption{Stability of the Cart--Stick Balancer is demonstrated using a sequentially decoded LTI code.}
\label{fig:angle}
\end{figure}
Although these simulations might suggest that a lower rate code always results in better stability of the dynamical system, it is not \emph{a-priori} clear that this is truly the case. A lower code rate uses a coarser quantizer than a higher rate code. As a result, there could be some loss due to this coarseness. A common metric used to quantify the performance of the closed-loop stability of a dynamical system is the linear quadratic regulator (LQR) cost for a finite time horizon $T$ given by 
\begin{equation}
J = \mathbb{E} \left[\frac{1}{2T} \sum_{t = 1}^{T}\left( \left\lVert \bx_t \right\rVert^2 + \left\lVert \bu_t \right\rVert^2 \right) \right] ,
\end{equation}
where the expectation is w.r.t.\ the randomness of the plant and the channel. For our example, we simulated the system using three different quantization levels, with 100 codes per quantization level and 40 experiments per code. 
The data is tabulated in Table \ref{tbl:LQG}.

\begin{table}[h!]
\caption{Average LQR Cost over 100 codes with 40 experiments per code}
\centering
\begin{tabular}{|c|c|}
\hline
$k$ & LQR Cost\\  \hline
4  & 206.0 \\ 
5  & 86.4 \\ 
10 & 873.0\\ \hline
\end{tabular}
\label{tbl:LQG}
\end{table}
In principle, one would randomly sample an LTI generator matrix where each subblock $\bG_i \in \mathbb{Z}_2^{n \times k}$. Nonetheless, from an implementation efficiency point of view, there is no loss in the anytime exponent if we pick $\bG_i \in \mathbb{Z}_2^{n' \times k'}$, where $n = n'\mathsf{gcd}(n,k)$ and $k = k'\mathsf{gcd}(n,k)$, and then using the first $t\mathsf{gcd}(n,k)$ blocks to encode $\{\bb_i\}_{i = 1}^t$.


\section{Discussion}
\label{s:summary}

We showed that sequential decoding algorithms have several desired features: 
Error exponential decay, memory that grows linearly (in contrast to the exponential growth under ML decoding) 
and expected complexity per branch that grows linearly (similarly to the encoding process of LTI tree codes).
However, the complexity distribution is heavy tailed (recall Theorems~\ref{thm:SeqDec:complexity:UB} and \ref{thm:SeqDec:complexity:LB}).
This means that there is a substantial probability that the computational complexity is going to be very large, 
which will cause, in turn, a failure in stabilizing the system. Specifically, by allowing only a finite backtracking length to the past, 
the computational complexity can be bounded at the expense of introducing an error due to failure.

From a practical point of view, the control specifications of the problem determine a probability of error threshold under which a branch $\bc_t$ is considered to be reliable. This can be used to set a limit on the delay-line length of the tree code, which in turn converts it to a \CC\ with a finite delay-line length. 

Finally, note that tighter bounds can be derived below the cutoff rate via expurgation. Moreover, 
random linear block codes are known to achieve the expurgated bound for block codes (with no need in expurgation)~\cite{BargForney}. 
Indeed, using this technique better bounds under ML decoding were derived in \cite{SukhavasiHassibi} 
and a similar improvement seems plausible for sequential decoding.


\section{Acknowledgments}

We thank R.~T.~Sukhavasi for many helpful discussions and the anonymous reviewers~--- for valuable comments.


\bibliographystyle{IEEEtran}

\end{document}